\documentclass[11pt,a4paper]{article}

\usepackage[utf8]{inputenc}
\usepackage[T1]{fontenc}
\usepackage[round,authoryear]{natbib}
\usepackage[english]{babel}
\usepackage{mathtools,amsfonts,amssymb,amsthm,mathrsfs}
\usepackage{cases}
\usepackage{enumerate}
\usepackage{tabularx}
\usepackage{xcolor}
\usepackage{graphicx}
\usepackage[left=2.7cm,right=2.7cm,top=2.7cm,bottom=3.4cm]{geometry}
\usepackage{url}
\usepackage{bm}
\usepackage{multirow}
\usepackage{authblk}
\usepackage[colorlinks=true,linkcolor=darkblue,citecolor=darkblue,urlcolor=darkblue]{hyperref}

\definecolor{darkblue}{rgb}{0,0,0.5}

\newcolumntype{Y}{>{\centering\arraybackslash}X}

\newtheorem{proposition}{Proposition}[section]
\newtheorem{theorem}{Theorem}[section]

\newtheorem{corollary}{Corollary}[section]

\theoremstyle{definition}
\newtheorem{definition}{Definition}[section]

\newtheorem{remark}{Remark}[section]

\newcommand{\reels}{\mathbb{R}}
\newcommand{\entiers}{\mathbb{N}}

\newcommand{\esp}{E}
\newcommand{\var}{V}

\newcommand{\argmax}{\operatorname*{arg \; max}}
\DeclareMathOperator{\e}{e}
\DeclareMathOperator{\dd}{\mbox{}}

\newcommand{\dx}{\dd{}\mathrm{d}x}

\title{\bfseries Comparative analysis and practical applications of cubic transmutations for the Pareto distribution}
\author[a]{Edoh Katchekpele\footnote{ edohkatchekpele@gmail.com}}
\author[b]{Issa Cherif Geraldo\footnote{cherifgera@gmail.com}}
\author[a]{Tchilabalo Abozou Kpanzou\footnote{kpanzout@gmail.com}}
\affil[a]{\footnotesize Laboratoire de Mod\'elisation Math\'ematique et d'Analyse Statistique Décisionnelle (LaMMASD), D\'epartement de Math\'ematiques, Facult\'e des Sciences et Techniques, Universit\'e de Kara, Kara, Togo.}
\affil[b]{\footnotesize Laboratoire d'Analyse, de Mod\'elisations Math\'ematiques et Applications (LAMMA), D\'epartement de Math\'ematiques, Facult\'e des Sciences, Université de Lomé, 1 B.P. 1515 Lom\'e 1, Togo.}

\date{}

\begin{document}
	
\renewcommand{\proofname}{\textbf{\normalshape Proof}}

\vspace{-2cm}
\maketitle

\begin{abstract}
\noindent
Transmutation is a technique for extending classical probability distributions in order to give them more flexibility. In this paper, we are interested in cubic transmutations of the Pareto distribution. We establish a general formula that unifies existing cubic transmutations of the Pareto distribution and facilitates the derivation of new cubic transmutations that have not yet been explored in the literature. We also derive general formulas for the related mathematical properties. Finally, we perform a comparative analysis of the six transmutations existing in the literature using real-world data. The results obtained confirm the flexibility and effectiveness of cubic transmutations in modeling various types of data.\\

\noindent 
\textbf{Keywords} Cubic transmutation; Pareto distribution; Parameter estimation; Real-world data analysis. \\
	
\noindent 
\textbf{AMS Subject Classification (MSC2020)} 60E05, 62E15, 62F10, 62F99, 62P99.
\end{abstract}

\section{Introduction}

\noindent
Cubic transmutations of probability distributions are widely recognized for their ability to increase the flexibility of classical distributions, enabling them to better model complex data. By introducing additional parameters, these transformations enhance the capacity to fit empirical data across various applications, including those involving heavy-tailed distributions. In particular, such methods are used to extend traditional distributions, like the Pareto distribution, to accommodate real-world datasets more effectively. This paper focuses on cubic transmutations applied to the Pareto distribution.\\

\noindent
Several studies have explored different transformations of the Pareto distribution. \cite{rahman-et-al-2020a} investigated the cubic transmuted Pareto distribution, analyzing its statistical properties, parameter estimation techniques, and performance through real data applications. Similarly, \cite{eledum-2020} introduced two generalized cubic transmuted versions of the Exponentiated Pareto-I distribution, demonstrating that one formulation (CTEP-I-G) outperforms the other (CTEP-I-R) in empirical studies. In contrast, \cite{monsef-et-al-2024} presented a three-parameter G-Transmuted Pareto distribution, providing an in-depth analysis of its mathematical properties and reliability characteristics. While their work highlights the effectiveness of transmutation in failure time modeling, it does not specifically address cubic transmutations.\\

\noindent
In a more recent study, \cite{geraldo-et-al-2025} compared various cubic transmutation approaches, particularly emphasizing the extension of classical probability distributions to better model complex data. The study explored relationships between different cubic transmutation methods, proposed modified versions by extending parameter ranges, and applied these modifications to the Pareto distribution as a baseline. However, this paper did not specifically address the six distinct cubic transmutations of the Pareto distribution, which is the focus of the current study.\\

\noindent
The main objective of this paper is to derive a general formula that encompasses all six cubic transmutations of the Pareto distribution, unifying the existing ones and introducing the three that have not yet been explored in the literature. We also perform a comparative analysis of these six transmutations, assessing their performance using real-world data to evaluate their practical applicability in various contexts.\\

\noindent
To achieve these objectives, the paper is structured as follows. Section \ref{sec:review-CT} provides a brief review of the key definitions and recent improvements in cubic transmutation formulas. Section \ref{sec:CTP} introduces a generalized formulation for cubic transmuted Pareto distributions. In Section \ref{sec4}, we discuss the statistical properties of these distributions, including moment calculations and other relevant results. The methodology for parameter estimation is presented in Section \ref{sec5}. Section \ref{sec:real-data} then applies the six cubic transmutations to real-world datasets, comparing their goodness-of-fit and practical relevance. Finally, Section \ref{sec:conclusion} summarizes the key findings and suggests directions for future research.

\section{A brief review on definitions and improvements of cubic transmutation formulas}  
\label{sec:review-CT}  

Let \( x \in \mathbb{R} \), and let \( G(x) \) be a cumulative distribution function (cdf). In general, \( G(x) \) may depend on a real parameter \( \xi \in \mathbb{R} \) or a parameter vector \( \xi \in \mathbb{R}^d \), where \( d \geqslant 2 \). However, for simplicity and without loss of generality, we will omit \( \xi \) in the notation.  

In what follows, we first recall the key definitions related to cubic transmutation formulas. Then, we discuss some recent improvements and refinements that enhance their applicability to probability distributions, particularly in the case of heavy-tailed models such as the Pareto distribution.

\subsection{Existing formulas}

In this subsection, we present six existing cubic transmutation formulas that have been proposed in the literature. These formulas define transformed cumulative distribution functions (cdfs) based on a given baseline cdf \( G(x) \). Each transmutation is characterized by specific parameter constraints that ensure the resulting function remains a valid cdf.  

The cubic transmutation due to \citet{granzotto-et-al-2017} (denoted \( CT_G \)) corresponds to the cdf  
\begin{equation}  
	\label{eq:cdf-granz}  
	F_G(x) = \lambda_1 G(x) + (\lambda_2 - \lambda_1) G^2(x) + (1 - \lambda_2) G^3(x),  
\end{equation}  
where the parameters \( \lambda_1 \) and \( \lambda_2 \) satisfy  
\begin{equation}  
	\label{eq:range-granz}  
	\mathscr{S}_G = \Big\{ (\lambda_1,\lambda_2) : 0 \leqslant \lambda_1 \leqslant 1 \;\; \text{and} \;\; -1 \leqslant \lambda_2 \leqslant 1 \Big\}.  
\end{equation}  

Similarly, the cubic transmutation of \citet{alkadim-mohammed-2017} (denoted \( CT_A \)) is given by the cdf  
\begin{equation}  
	\label{eq:cdf-alkdm}  
	F_A(x) = (1+\lambda) G(x) - 2\lambda G^2(x) + \lambda G^3(x),  
\end{equation}  
where \( -1 \leqslant \lambda \leqslant 1 \).  

\citet{rahman-et-al-2018a} proposed another cubic transmutation formula, whose cdf is  
\begin{equation}  
	\label{eq:cdf-r18a}  
	F_{R18a}(x) = (1 + \lambda_1) G(x) + (\lambda_2 - \lambda_1) G^2(x) - \lambda_2 G^3(x),  
\end{equation}  
with parameters constrained by  
\begin{equation}  
	\label{eq:range-r18a}  
	\mathscr{S}_{R18a} = \Big\{ (\lambda_1,\lambda_2) :-1 \leqslant \lambda_1 \leqslant 1, \;\; -1 \leqslant \lambda_2 \leqslant 1 \;\; \text{and} \;\; -2 \leqslant \lambda_1 + \lambda_2 \leqslant 1 \Big\}.  
\end{equation}  

Building on this approach, \citet{rahman-et-al-2018b} introduced a new transformation of the form  
\begin{equation}  
	\label{eq:cdf-r18b}  
	F_{R18b}(x) = (1 + \lambda_1 + \lambda_2) G(x) - (\lambda_1 + 2\lambda_2) G^2(x) + \lambda_2 G^3(x),  
\end{equation}  
where the parameters are constrained by  
\begin{equation}  
	\label{eq:range-r18b}  
	\mathscr{S}_{R18b} = \Big\{ (\lambda_1,\lambda_2) : -1 \leqslant \lambda_1 \leqslant 1 \; \text{and} \; 0 \leqslant \lambda_2 \leqslant 1 \Big\}.  
\end{equation}  

Later, \citet{rahman-et-al-2019b} proposed another cubic transmutation given by the cdf  
\begin{equation}  
	\label{eq:cdf-r19}  
	F_{R19}(x) = (1 - \lambda) G(x) + 3\lambda G^2(x) - 2\lambda G^3(x),  
\end{equation}  
where \( -1 \leqslant \lambda \leqslant 1 \).  

More recently, \citet{rahman-et-al-2023} developed a modified cubic transmuted family of distributions, whose cdf is given by  
\begin{equation}  
	\label{eq:cdf-r23}  
	F_{R23}(x) = \left[ 1 - \lambda (\eta - 1) \right] G(x) + \lambda (2\eta - 1) G^2(x) - \lambda \eta G^3(x),  
\end{equation}  
where  
\begin{equation}  
	\label{eq:range-r23}  
	(\lambda,\eta) \in \mathscr{S}_{R23} = \Big\{ (\lambda, \eta) : -1 \leqslant \lambda \leqslant 1 \;\; \text{and} \;\; 0 \leqslant \eta \leqslant 2 \Big\}.  
\end{equation}  

While these cubic transmutation formulas provide flexible extensions of the baseline distribution, there is still room for improvement in terms of generalization and practical applicability. In the next section, we explore possible enhancements to these transformations.

\subsection{Improvements of cubic transmuted distributions}
Let \( G(x) \) be the cumulative distribution function (cdf) of a given random variable. \citet{geraldo-et-al-2025} have made theoretical comparisons of different cubic transmutation formulas with baseline cdf \( G(x) \) and proposed modified versions based on new parameter ranges. They also demonstrated that some of these distribution families are subfamilies of others.\\
Their main findings are summarized hereafter.

\begin{list}{$\bullet$}{}
\item
The Granzotto Cubic Transmuted distribution with cdf given by Equation \eqref{eq:cdf-granz} is not always a well-defined distribution when parameters $\lambda_1$ and $\lambda_2$ satisfy Equation \eqref{eq:range-granz}.

The revised set of conditions
	\begin{equation} 
		\label{eq:range-granz-b} 
		\mathscr{S}_{MG} = \Big\{ (\lambda_1,\lambda_2) : 0 \leqslant \lambda_1 \leqslant 3, \; 0 \leqslant \lambda_2 \leqslant 3, \; \text{and} \; 0 \leqslant \lambda_1 + \lambda_2 \leqslant 3 \Big\}.
	\end{equation}
ensures the validity of $CT_G$ leading to the modified Granzotto Cubic Transmutation
denoted \( CT_{MG} \).

	\item 
	The cubic transmuted family $ CT_A$ associated with $G(x)$, is well defined under the extended condition \( -1 \leqslant \lambda \leqslant 3 \) and this revised version is referred to as the modified Al-Kadim Cubic Transmutation and denoted $CT_{M A}$.
	
	\item 
	The cubic transmuted distribution denoted $CT_{R18a}$, proposed by \citet{rahman-et-al-2018a} and originally defined under conditions \eqref{eq:range-r18a}  is well-defined. However, the extended parameter range
	\begin{equation} 
		\label{eq:range-r18a-b}
		\mathscr{S}_{MR18a} = \Big\{ (\lambda_1,\lambda_2) : -1 \leqslant \lambda_1 \leqslant 2, \; -1 \leqslant \lambda_2 \leqslant 2, \; \text{and} \; -2 \leqslant \lambda_1 + \lambda_2 \leqslant 1 \Big\}
	\end{equation}
	ensures a broader applicability and leads to the modified version denoted \( CT_{MR18a} \).
	
	\item Similarly, the cubic transmuted distribution introduced by \citet{rahman-et-al-2018b}, denoted \( CT_{R18b}$ and originally constrained by conditions \eqref{eq:range-r18b} is more flexible under the more general parameter range given by
	\begin{equation} 
		\label{eq:range-r18b-b}
		\mathscr{S}_{MR18b} = \Big\{ (\lambda_1,\lambda_2) : -2 \leqslant \lambda_1 \leqslant 1, \; -2 \leqslant \lambda_2 \leqslant 1, \; \text{and} \; -1 \leqslant \lambda_1 + \lambda_2 \leqslant 2 \Big\},
	\end{equation}
	thus leading to the modified version denoted \( CT_{MR18b} \).
	
	\item The modified families \( CT_{MG} \), \( CT_{MR18a} \), and \( CT_{MR18b} \) have been found to be equivalent.
	
	\item The cubic transmuted distribution introduced by \citet{rahman-et-al-2019b}, originally denoted $CT_{R19}$, remains well defined under the extended range  \( -2 \leqslant \lambda \leqslant 1 \) and this leads to the modified cubic transmutation, denoted \( CT_{MR19} \).
\end{list}

The modifications introduced by \citet{geraldo-et-al-2025} and presented above, extend
the applicability of cubic transmuted distributions by redefining their parameter ranges while preserving their flexibility. Building on these refined formulations, we now introduce a general definition of cubic transmutation that unifies these different approaches.

\subsection{A proposed definition of cubic transmutation}



To the best of our knowledge, and as reported in the literature \citep[see, for example,][]{rahman-et-al-2020b,ali-athar-2021,imliyangba-et-al-2021}, cubic transmutations of a cumulative distribution function (cdf) \( G(x) \) follow the general form:
\begin{equation}
	F(x) = \delta_1 G(x) + \delta_2 G^2(x) + \delta_3 G^3(x),
\end{equation}
where \( \delta = (\delta_1, \delta_2, \delta_3) \in \mathbb{R}^3 \) is a parameter vector, with each \( \delta_i \) (\( i=1,2,3 \)) potentially depending on one or more real parameters. The function \( F(x) \) must also be a valid cdf, meaning that it satisfies the following conditions:  
\begin{subnumcases}{}
	\label{eq:def-cdf-a} \lim_{x \to -\infty} F(x) = 0, \\
	\label{eq:def-cdf-b} \lim_{x \to +\infty} F(x) = 1, \\
	\label{eq:def-cdf-c} \forall x \in \mathbb{R}, \quad F'(x) \geq 0.
\end{subnumcases}

Since \( G(x) \) is itself a cdf, condition \eqref{eq:def-cdf-a} is always satisfied. Additionally, condition \eqref{eq:def-cdf-b} leads to the constraint:
\begin{equation} 
	\delta_1 + \delta_2 + \delta_3 = 1.
\end{equation}

To examine condition \eqref{eq:def-cdf-c}, let \( g(x) = \frac{dG(x)}{dx} \) be the probability density function (pdf) associated with \( G(x) \), and define the pdf of \( F(x) \) as:
$$
	f(x) = \frac{dF(x)}{dx} = g(x) \left[ \delta_1 + 2\delta_2 G(x) + 3\delta_3 G^2(x) \right].
$$

Thus, condition \eqref{eq:def-cdf-c} is equivalent to requiring that the function \( r(t) = \delta_1 + 2\delta_2 t + 3\delta_3 t^2 \), defined for all \( t = G(x) \in [0,1] \), remains non-negative. This observation leads to the following formal definition of cubic transmutation.

\begin{definition}
	\label{def:CT}
	Let $x \in \reels$ and $G(x)$ be a cumulative density function (cdf). A cdf is a cubic transmutation of $G(x)$ if it has the form
	\begin{equation}
	\label{eqCTN}
		F(x) = \delta_1 G(x) + \delta_2 G^2(x) + (1-\delta_1-\delta_2) G^3(x),
	\end{equation}
	where $(\delta_1,\delta_2) \in \reels^2$ is such that 
	\begin{equation}
		\label{eq:delta-1} 
		\forall t \in [0,1], \quad r(t) = \delta_1 + 2\delta_2 t + 3(1-\delta_1-\delta_2) t^2 \geqslant 0
	\end{equation}
	or
	\begin{equation}
		\label{eq:delta-2} 
		\inf_{t \in [0,1]} \Big( \delta_1 + 2\delta_2 t + 3(1-\delta_1-\delta_2) t^2 \Big) \geqslant 0.
	\end{equation}
\end{definition}

%

\begin{remark}
	Note that, in Definition \ref{def:CT}, we have  
	\begin{equation*}
		\int_{0}^{1} r(t) \, dt = \left[ \delta_1 t + \delta_2 t^2 + \delta_3 t^3 \right]_0^1 = 1.
	\end{equation*}
	Thus, the assumption that \( r(t) \) is positive for all \( t \in [0,1] \) is equivalent to stating that \( r(t) \) is a valid probability density function (pdf) on \( [0,1] \). Moreover, if \( R(t) = \delta_1 t + \delta_2 t^2 + \delta_3 t^3 \) is the cdf corresponding to the pdf \( r(t) \), then  
	\begin{equation}
		\label{eq:alzaatreh}
		F(x) = R(G(x)) = \int_{0}^{G(x)} r(t) \, dt.
	\end{equation}
	This corresponds to a special case of a general formula proposed by \citet{alzaatreh-et-al-2013}.
\end{remark}

\section{A generalized formula for the Cubic Transmutation of Pareto distributions}
\label{sec:CTP}

The Pareto distribution is defined by its cdf as 
\begin{equation}
	G(x) = 1 - \left(\frac{x_0}{x} \right)^{\alpha}, \quad x \geqslant x_0,
\end{equation}
and its pdf as
\begin{equation}
	g(x) = \frac{\alpha x_0^{\alpha}}{x^{\alpha+1}}, \quad x \geqslant x_0,
\end{equation}
where $x_0 \in \mathbb{R}_{+}^{*}$ and $\alpha \in \mathbb{R}_{+}^{*}$. 
	
%

Based on the quadratic transmutation proposed by \citet{shaw-buckley-2009}, \citet{merovci-puka-2014} developed the transmuted Pareto (TP) distribution. Several authors have since investigated cubic transmuted Pareto (CTP) distributions.  

Building on the cubic transmutation (CT) formula of \citet{alkadim-mohammed-2017}, \citet{ansari-eledum-2018} introduced a CTP distribution, which we denote as \( CTP_A \). \citet{rahman-et-al-2020a} proposed another CTP distribution, denoted as \( CTP_{R18a} \), based on the formula of \citet{rahman-et-al-2018a}. Similarly, \citet{rahman-et-al-2021} introduced yet another CTP distribution, denoted as \( CTP_{R18b} \), using the formula of \citet{rahman-et-al-2018b}.  

To the best of our knowledge, other CT formulas have not yet been applied to the Pareto distribution. We denote the CTP distributions derived from the formulas of \citet{granzotto-et-al-2017}, \citet{rahman-et-al-2019b}, and \citet{rahman-et-al-2023} as \( CTP_G \), \( CTP_{R19} \), and \( CTP_{R23} \), respectively.  

Furthermore, the modified versions of \( CTP_{A} \), \( CTP_{R18a} \), \( CTP_{R18b} \), \( CTP_{G} \), and \( CTP_{R19} \) will be referred to as \( CTP_{MA} \), \( CTP_{MR18a} \), \( CTP_{MR18b} \), \( CTP_{MG} \), and \( CTP_{MR19} \), respectively.  

The following proposition provides a generalized formula for the cumulative distribution function (cdf) and probability density function (pdf) of the CTP distribution.

\begin{proposition} 
	The cumulative distribution function (cdf) and the probability density function (pdf) of the cubic transmutation of the Pareto distribution are respectively given by
	\begin{equation}
		\label{eq:cdf-CTP}
		F(x) = 1 + (2\delta_1 + \delta_2 - 3) \left( \frac{x_0}{x} \right)^{\alpha} + (3 - 3\delta_1 - 2\delta_2) \left( \frac{x_0}{x} \right)^{2\alpha} + (\delta_1 + \delta_2 - 1) \left( \frac{x_0}{x} \right)^{3\alpha}
	\end{equation}
	and
	\begin{equation}
		\label{eq:pdf-CTP}
		f(x) = \frac{\alpha x_0^{\alpha}}{x^{\alpha+1}} \left[ (3 - 2\delta_1 - \delta_2) + (6\delta_1 + 4\delta_2 - 6) \left( \frac{x_0}{x} \right)^{\alpha} + (3 - 3\delta_1 - 3\delta_2) \left( \frac{x_0}{x} \right)^{2\alpha} \right],
	\end{equation}
	where $x_0 > 0$, $x \geqslant x_0$, $\alpha > 0$ and $\delta_1$ and $\delta_2$ satisfy the conditions \eqref{eq:delta-1} or \eqref{eq:delta-2} so that $f$ is indeed a true pdf.
\end{proposition}

\begin{proof}
	Let $F$ be the cdf and $f$ the pdf of any cubic transmutation of Pareto distribution. We have, using Equation \eqref{eqCTN},
	\begin{align*}
		F(x) & =  \delta_1 \left( 1 - \left( \frac{x_0}{x} \right)^{\alpha} \right) + \delta_2 \left( 1 - \left( \frac{x_0}{x} \right)^{\alpha} \right)^2 + (1-\delta_1-\delta_2) \left( 1 - \left( \frac{x_0}{x} \right)^{\alpha} \right)^3 \\
		& = \delta_1 \left( 1 - \left( \frac{x_0}{x} \right)^{\alpha} \right) + \delta_2 \left( 1 - 2\left( \frac{x_0}{x} \right)^{\alpha} + \left( \frac{x_0}{x} \right)^{2\alpha} \right) \\
		& \qquad\qquad + (1-\delta_1-\delta_2) \left( 1 - 3\left( \frac{x_0}{x} \right)^{\alpha} + 3\left( \frac{x_0}{x} \right)^{2\alpha} - \left( \frac{x_0}{x} \right)^{3\alpha}\right) \\
		& = \delta_1 - \delta_1 \left( \frac{x_0}{x} \right)^{\alpha} + \delta_2 - 2\delta_2 \left( \frac{x_0}{x} \right)^{\alpha} + \delta_2 \left( \frac{x_0}{x} \right)^{2\alpha} \\
		& \quad + (1-\delta_1-\delta_2) - 3(1-\delta_1-\delta_2) \left( \frac{x_0}{x} \right)^{\alpha} + 3(1-\delta_1-\delta_2) \left( \frac{x_0}{x} \right)^{2\alpha} - (1-\delta_1-\delta_2) \left( \frac{x_0}{x} \right)^{3\alpha} \\
		& = (\delta_1 + \delta_2 + 1-\delta_1-\delta_2) + (- \delta_1 - 2\delta_2 - 3 + 3\delta_1 + 3\delta_2) \left( \frac{x_0}{x} \right)^{\alpha} \\
		& \qquad\qquad + (\delta_2 + 3 - 3\delta_1 - 3\delta_2) \left( \frac{x_0}{x} \right)^{2\alpha} - (1-\delta_1-\delta_2) \left( \frac{x_0}{x} \right)^{3\alpha} \\
		& = 1 + (2\delta_1 + \delta_2 - 3) \left( \frac{x_0}{x} \right)^{\alpha} + (3 - 3\delta_1 - 2\delta_2) \left( \frac{x_0}{x} \right)^{2\alpha} + (\delta_1 + \delta_2 - 1) \left( \frac{x_0}{x} \right)^{3\alpha}.
	\end{align*}
	Since $f(x) = F'(x)$, we also have
	\begin{align*}
		f(x) & = (3 - 2\delta_1 - \delta_2) \frac{\alpha x_0}{x^2} \left( \frac{x_0}{x} \right)^{\alpha-1} + (6\delta_1 + 4\delta_2 - 6) \frac{\alpha x_0}{x^2} \left( \frac{x_0}{x} \right)^{2\alpha-1} \\ 
		& \qquad\qquad + (3 - 3\delta_1 - 3\delta_2) \frac{\alpha x_0}{x^2} \left( \frac{x_0}{x} \right)^{3\alpha-1} \\
		f(x) & = \frac{\alpha x_0^{\alpha}}{x^{\alpha+1}} \left[ (3 - 2\delta_1 - \delta_2) + (6\delta_1 + 4\delta_2 - 6) \left( \frac{x_0}{x} \right)^{\alpha} + (3 - 3\delta_1 - 3\delta_2) \left( \frac{x_0}{x} \right)^{2\alpha} \right].
	\end{align*}
\end{proof} 

\begin{remark}
	\label{rem:rem-CTP1}
By applying Equations \eqref{eq:cdf-CTP} and \eqref{eq:pdf-CTP} for different values of \( \delta_1 \) and \( \delta_2 \), we can verify the cdf and pdf of the cubic-transmuted Pareto distributions found in the literature. 

		\begin{list}{$\bullet$}{}
		\item Setting $\delta_1 = 1+\lambda$ and $\delta_2 = -2\lambda$ yields the cdf and the pdf of $CTP_A$ \citep{ansari-eledum-2018} and its modification $CTP_{MA}$ proposed in this paper:
		\begin{equation}
			F_A(x) = 1 - \left( \frac{x_0}{x} \right)^{\alpha} + \lambda \left( \frac{x_0}{x} \right)^{2\alpha} - \lambda \left( \frac{x_0}{x} \right)^{3\alpha}
		\end{equation}
		and
		\begin{equation}
			f_A(x) = \frac{\alpha x_0^{\alpha}}{x^{\alpha + 1}} \left[ 1 - 2\lambda \left( \frac{x_0}{x} \right)^{\alpha} + 3\lambda \left( \frac{x_0}{x} \right)^{2\alpha} \right].
		\end{equation}
		
		\item By setting $\delta_1 = 1+\lambda_1$ and $\delta_2 = \lambda_2 - \lambda_1$, we respectively obtain the cdf and the pdf of $CTP_{R18a}$ \citep{rahman-et-al-2020a} and its modification $CTP_{MR18a}$:
		\begin{equation}
			F_{R18a}(x) = 1 + (\lambda_1 + \lambda_2 - 1) \left( \frac{x_0}{x} \right)^{\alpha} - (\lambda_1 + 2\lambda_2) \left( \frac{x_0}{x} \right)^{2\alpha} + \lambda_2 \left( \frac{x_0}{x} \right)^{3\alpha}
		\end{equation}
		and
		\begin{equation}
			f_{R18a}(x) = \frac{\alpha x_0^{\alpha}}{x^{\alpha + 1}} \left[ (1 - \lambda_1 - \lambda_2) + 2(\lambda_1 + 2\lambda_2) \left( \frac{x_0}{x} \right)^{\alpha} - 3\lambda_2 \left( \frac{x_0}{x} \right)^{2\alpha} \right].
		\end{equation}
		
		\item Setting $\delta_1 = 1 + \lambda_1 + \lambda_2$ and $\delta_2 = -\lambda_1 - 2\lambda_2$, we derive the cdf and pdf of $CTP_{R18b}$  \citep{rahman-et-al-2021} and its modified version $CTP_{MR18b}$: 
		\begin{equation}
			F_{R18b}(x) = 1 + (\lambda_1 - 1) \left( \frac{x_0}{x} \right)^{\alpha} + (\lambda_2 - \lambda_1) \left( \frac{x_0}{x} \right)^{2\alpha} - \lambda_2 \left( \frac{x_0}{x} \right)^{3\alpha}
		\end{equation}
		and
		\begin{equation}
			f_{R18b}(x) = \frac{\alpha x_0^{\alpha}}{x^{\alpha + 1}} \left[ (1 - \lambda_1) + 2(\lambda_1 - \lambda_2) \left( \frac{x_0}{x} \right)^{\alpha} + 3\lambda_2 \left( \frac{x_0}{x} \right)^{2\alpha} \right].
		\end{equation}
	\end{list}
\end{remark}
%

The following corollaries provide the respective cdf and pdf of the CTP, which are novel and have not been addressed in the literature previously.

\begin{corollary}
	For $\delta_1 = \lambda_1$ and $\delta_2 = \lambda_2 - \lambda_1$, we respectively get the cdf and the pdf of both $CTP_G$ and $CTP_{MG}$: 
	\begin{equation}
		F_G(x) = 1 + (\lambda_1 + \lambda_2 - 3) \left( \frac{x_0}{x} \right)^{\alpha} + (3 - \lambda_1 - 2\lambda_2) \left( \frac{x_0}{x} \right)^{2\alpha} + (\lambda_2 - 1) \left( \frac{x_0}{x} \right)^{3\alpha}
	\end{equation}
	and
	\begin{equation}
		f_G(x) = \frac{\alpha x_0^{\alpha}}{x^{\alpha + 1}} \left[ (3 - \lambda_1 - \lambda_2) + 2(\lambda_1 + 2\lambda_2 - 3) \left( \frac{x_0}{x} \right)^{\alpha} + 3(1 - \lambda_2) \left( \frac{x_0}{x} \right)^{2\alpha} \right].
	\end{equation}
\end{corollary}

\begin{corollary} 
	Setting $\delta_1 = 1-\lambda$ and $\delta_2 = 3\lambda$, we respectively get the cdf and pdf of the $CTP_{R19}$ and its modified version $CTP_{MR19}$:
	\begin{equation}
		F_{R19}(x) = 1 + (\lambda - 1) \left( \frac{x_0}{x} \right)^{\alpha} - 3\lambda \left( \frac{x_0}{x} \right)^{2\alpha} + 2\lambda \left( \frac{x_0}{x} \right)^{3\alpha}
	\end{equation}
	and
	\begin{equation}
		f_{R19}(x) = \frac{\alpha x_0^{\alpha}}{x^{\alpha + 1}} \left[ (1-\lambda) + 6\lambda \left( \frac{x_0}{x} \right)^{\alpha} - 6\lambda \left( \frac{x_0}{x} \right)^{2\alpha} \right].
	\end{equation}
\end{corollary}

\begin{corollary} 	
	Setting $\delta_1 = 1 + \lambda - \lambda \eta$ and $\delta_2 = 2 \lambda \eta - \lambda$, we respectively get the cdf and pdf of the $CTP_{R23}$:
	\begin{equation}
		F_{R23}(x) = 1 + (\lambda - 1) \left( \frac{x_0}{x} \right)^{\alpha} - \lambda (1 + \eta) \left( \frac{x_0}{x} \right)^{\alpha} + \lambda \eta \left( \frac{x_0}{x} \right)^{\alpha}
	\end{equation}
	and
	\begin{equation}
		f_{R23}(x) = \frac{\alpha x_0^{\alpha}}{x^{\alpha + 1}} \left[ (1 - \lambda) + 2\lambda (1 + \eta) \left( \frac{x_0}{x} \right)^{\alpha} - 3 \lambda \eta \left( \frac{x_0}{x} \right)^{2\alpha} \right].
	\end{equation}
\end{corollary}

\section{Moments and other results}
\label{sec4}

\subsection{Moments}

\begin{theorem}
	Let $X$ be a random variable following the CTP distribution. For all $k \in \entiers^*$, the $k^{\text{th}}$ moment of $X$ is given by 
	\begin{equation} 
		\label{eq:mom-CTP}
		\esp(X^k) = \alpha x_0^k \left[ \frac{-\delta_1 k^2 + (5\delta_1 + 2\delta_2) k\alpha - 6\alpha^2}{(k-\alpha)(k-2\alpha)(k-3\alpha)} \right], \quad \alpha > k.
	\end{equation}
	The mean and variance are respectively given by 
	\begin{equation} 
		\esp(X) = \alpha x_0 \left[ \frac{-\delta_1 + (5\delta_1 + 2\delta_2) \alpha - 6\alpha^2}{(1-\alpha)(1-2\alpha)(1-3\alpha)} \right], \quad \alpha > 1
	\end{equation}
	and
	\begin{multline} 
		\var(X) = \alpha x_0^2 \Bigg[ \frac{-4\delta_1 + (10\delta_1 + 4\delta_2) \alpha - 6\alpha^2}{(2-\alpha)(2-2\alpha)(2-3\alpha)} \\
		- \alpha \left( \frac{-\delta_1 + (5\delta_1 + 2\delta_2) \alpha - 6\alpha^2}{(1-\alpha)(1-2\alpha)(1-3\alpha)} \right)^2 \Bigg], \qquad \alpha > 2
	\end{multline}
\end{theorem}

\begin{proof}
	By definition, the moment of order $k$ is defined as
	\begin{equation} 
		\esp(X^k) = \int_{-\infty}^{\infty} x^k f(x) \dx.
	\end{equation}
	By applying this formula to the pdf $f(x)$ of any CTP distribution defined by Equation \eqref{eq:pdf-CTP}, we have:
	\begin{align*}
		\esp(X^k) & = \int_{x_0}^{\infty} x^k \frac{\alpha x_0^{\alpha}}{x^{\alpha+1}} \left[ (3 - 2\delta_1 - \delta_2) + (6\delta_1 + 4\delta_2 - 6) \left( \frac{x_0}{x} \right)^{\alpha} + (3 - 3\delta_1 - 3\delta_2) \left( \frac{x_0}{x} \right)^{2\alpha} \right] \dx \\
		& = (3 - 2\delta_1 - \delta_2) \int_{x_0}^{\infty} \frac{\alpha x_0^{\alpha}}{x^{\alpha+1-k}} \dx + (6\delta_1 + 4\delta_2 - 6) \int_{x_0}^{\infty} \frac{\alpha x_0^{2\alpha}}{x^{2\alpha+1-k}} \dx \\
		& \qquad + (3 - 3\delta_1 - 3\delta_2) \int_{x_0}^{\infty} \frac{\alpha x_0^{3\alpha}}{x^{3\alpha+1-k}} \dx \\
		& = \alpha x_0^{\alpha} (3 - 2\delta_1 - \delta_2) \left[ -\frac{1}{(\alpha-k)x^{\alpha-k}} \right]_{x_0}^{\infty} + \alpha x_0^{2\alpha} (6\delta_1 + 4\delta_2 - 6) \left[ -\frac{1}{(2\alpha-k) x^{2\alpha-k}} \right]_{x_0}^{\infty} \\
		& \qquad + \alpha x_0^{3\alpha} (3 - 3\delta_1 - 3\delta_2) \left[ -\frac{1}{(3\alpha-k) x^{3\alpha-k}} \right]_{x_0}^{\infty}.
	\end{align*}
	The limits in $\infty$ involved are finite (and zero) if we simultaneously have $\alpha > k$, $2\alpha > k$ and $3\alpha > k$. Since $\alpha > 0$, this is equivalent to $\alpha > k$. So, if $\alpha > k$,
	\begin{align*}
		\esp(X^k) & = \frac{\alpha x_0^{\alpha} (3 - 2\delta_1 - \delta_2)}{(\alpha-k) x_0^{\alpha-k}} + \frac{\alpha x_0^{2\alpha} (6\delta_1 + 4\delta_2 - 6)}{(2\alpha-k) x_0^{2\alpha-k}} + \frac{\alpha x_0^{3\alpha} (3 - 3\delta_1 - 3\delta_2)}{(3\alpha-k) x_0^{3\alpha-k}} \\
		& = \alpha x_0^k \left[ \frac{2\delta_1 + \delta_2 - 3}{k-\alpha} + \frac{6 - 6\delta_1 - 4\delta_2}{k-2\alpha} + \frac{3\delta_1 + 3\delta_2 - 3}{k-3\alpha} \right] \\
		& = \alpha x_0^k \left[ \frac{(2\delta_1 + \delta_2 - 3) (k-2\alpha) (k-3\alpha) + (6 - 6\delta_1 - 4\delta_2) (k-\alpha)(k-3\alpha)}{(k-\alpha)(k-2\alpha)(k-3\alpha)} \right. \\
		& \qquad\qquad \left. + \frac{(3\delta_1 + 3\delta_2 - 3) (k-\alpha) (k-2\alpha)}{(k-\alpha)(k-2\alpha)(k-3\alpha)} \right] \\
		& = \alpha x_0^k \left[ \frac{(2\delta_1 + \delta_2 - 3) (k^2 - 5k\alpha + 6\alpha^2) + (6 - 6\delta_1 - 4\delta_2) (k^2 - 4k\alpha + 3\alpha^2)}{(k-\alpha)(k-2\alpha)(k-3\alpha)} \right. \\
		& \qquad\qquad \left. + \frac{(3\delta_1 + 3\delta_2 - 3) (k^2 - 3k\alpha + 2\alpha^2)}{(k-\alpha)(k-2\alpha)(k-3\alpha)} \right] \\
		& = \alpha x_0^k \left[ \frac{(2\delta_1 + \delta_2 - 3) k^2 + (-10\delta_1 - 5\delta_2 + 15) k\alpha + (12\delta_1 + 6\delta_2 - 18) \alpha^2}{(k-\alpha)(k-2\alpha)(k-3\alpha)} \right. \\
		& \qquad\qquad + \frac{(6 - 6\delta_1 - 4\delta_2) k^2 + (-24 + 24\delta_1 + 16\delta_2) k\alpha + (18 - 18\delta_1 - 12\delta_2) \alpha^2}{(k-\alpha)(k-2\alpha)(k-3\alpha)} \\
		& \qquad\qquad \left. + \frac{(3\delta_1 + 3\delta_2 - 3) k^2 + (-9\delta_1 - 9\delta_2 + 9) k\alpha + (6\delta_1 + 6\delta_2 - 6) \alpha^2}{(k-\alpha)(k-2\alpha)(k-3\alpha)} \right] \\
		& = \alpha x_0^k \left[ \frac{-\delta_1 k^2 + (5\delta_1 + 2\delta_2) k\alpha - 6\alpha^2}{(k-\alpha)(k-2\alpha)(k-3\alpha)} \right].
	\end{align*}
	The mathematical expectation $\esp(X)$ and the variance $\var(X)$ are easily deduced by applying Equation \eqref{eq:mom-CTP} for $k=1$, $k=2$ and by using the relation $\var(X) = \esp(X^2) - \left[ \esp(X) \right]^2$.
\end{proof}

\begin{remark}
	By applying Equation \eqref{eq:mom-CTP} for given values of $\delta_1$ and $\delta_2$, we get the moments of the different cubic transmuted Pareto distributions. 
	\begin{list}{$\bullet$}{}		
		\item Setting $\delta_1 = 1+\lambda$ and $\delta_2 = -2\lambda$, we get the moment of order $k$ of $CTP_{A}$ \citep{ansari-eledum-2018} and its modified version $CTP_{MA}$:
		\begin{equation} 
			\esp(X^k) = \alpha x_0^k \left[ \frac{-(1+\lambda) k^2 + (5+\lambda) k\alpha - 6\alpha^2}{(k-\alpha)(k-2\alpha)(k-3\alpha)} \right], \quad \alpha > k.
		\end{equation}
		The result is identical to that found in \cite[Proposition 3.1, p. 447]{ansari-eledum-2018} for $CTP_{A}$ by considering the notations of these latter authors.
		
		\item Setting $\delta_1 = 1+\lambda_1$ and $\delta_2 = \lambda_2 - \lambda_1$, we get the moment of order $k$ of the $CTP_{R18a}$ \citep{rahman-et-al-2020a} and its modified version:
		\begin{equation} 
			\esp(X^k) = \alpha x_0^k \left[ \frac{-(1+\lambda_1) k^2 + (5 + 3\lambda_1 + 2\lambda_2) k\alpha - 6\alpha^2}{(k-\alpha)(k-2\alpha)(k-3\alpha)} \right], \quad \alpha > k.
		\end{equation}
		The result is identical to that found in \cite[Theorem 1]{rahman-et-al-2020a} for $CTP_{R18a}$.
		
		\item Setting $\delta_1 = 1 + \lambda_1 + \lambda_2$ and $\delta_2 = -\lambda_1 - 2\lambda_2$, we get the moment of order $k$ of the $CTP_{R18b}$ \citep{rahman-et-al-2021} and its modified version: 
		\begin{equation} 
			\esp(X^k) = \alpha x_0^k \left[ \frac{-(1 + \lambda_1 + \lambda_2) k^2 + (5 + 3\lambda_1 + \lambda_2) k\alpha - 6\alpha^2}{(k-\alpha)(k-2\alpha)(k-3\alpha)} \right], \quad \alpha > k.
		\end{equation}
		The result is identical to that found in \cite[Theorem 1, p. 645]{rahman-et-al-2021} for $CTP_{R18b}$.
	\end{list}
\end{remark}

The following corollaries provide the $k^{\text{th}}$ order moments of the CTP, which are novel and have not been previously addressed in the literature.

\begin{corollary}
Setting $\delta_1 = \lambda_1$ and $\delta_2 = \lambda_2 - \lambda_1$, we obtain the moment of order $k$ of both $CTP_{G}$ and $CTP_{MG}$ as: 
		\begin{equation} 
			\esp(X^k) = \alpha x_0^k \left[ \frac{-\lambda_1 k^2 + (3\lambda_1 + 2\lambda_2) k\alpha - 6\alpha^2}{(k-\alpha)(k-2\alpha)(k-3\alpha)} \right], \quad \alpha > k.
		\end{equation}
\end{corollary}

\begin{corollary} 
Setting $\delta_1 = 1-\lambda$ and $\delta_2 = 3\lambda$, we obtain the $k^{\text{th}}$ order moment of both $CTP_{R19}$ and $CTP_{MR19}$:
		\begin{equation} 
			\esp(X^k) = \alpha x_0^k \left[ \frac{(\lambda-1) k^2 + (5 + \lambda) k\alpha - 6\alpha^2}{(k-\alpha)(k-2\alpha)(k-3\alpha)} \right], \quad \alpha > k.
		\end{equation}
\end{corollary}	
	
\begin{corollary} 
 Setting $\delta_1 = 1 + \lambda - \lambda \eta$ and $\delta_2 = 2 \lambda \eta - \lambda$, we obtain the $k^{\text{th}}$ order moment of the $CTP_{R23}$:
		\begin{equation} 
			\esp(X^k) = \alpha x_0^k \left[ \frac{-(1 + \lambda - \lambda \eta) k^2 + (5 + 3\lambda - \lambda \eta) k\alpha - 6\alpha^2}{(k-\alpha)(k-2\alpha)(k-3\alpha)} \right], \quad \alpha > k.
		\end{equation}
\end{corollary}		
		
\subsection{Moment generating function}

\begin{theorem}
	Let $X$ be a random variable following the CTP distribution. The moment generating function is defined for all $t \in \mathbb{R}$ by
	\begin{equation} 
		\label{eq:mgf-CTP}
		M_X(t) = \sum_{k=0}^{\infty} \frac{\alpha t^k x_0^k \left[ -\delta_1 k^2 + (5\delta_1 + 2\delta_2) k\alpha - 6\alpha^2 \right]}{k! \, (k-\alpha)(k-2\alpha)(k-3\alpha)},
	\end{equation}
	where $\alpha>k$.
\end{theorem}

\begin{proof}
	The moment generating function of a v.a $X$ is given by
	$$
	M_X(t) = \esp \left( \e^{tX} \right) = \int_0^{\infty} \e^{tx} f(x) \dx.
	$$
	From the power series expansion of $\e^{tx}$, we have
	\begin{align*} 
		M_X(t) & = \int_0^{\infty} \left( \sum_{k=0}^{\infty} \frac{(tx)^k}{k!} \right) f(x) \dx \\
		& = \int_0^{\infty} \left( f(x) + tx f(x) + \frac{t^2 x^2}{2!} f(x) + \cdots + \frac{t^n x^n}{n!} f(x) + \cdots \right) \dx \\
		& = \sum_{k=0}^{\infty} \frac{t^k \esp(X^k)}{k!}.
	\end{align*}
\end{proof}

\subsection{Characteristic function}

\begin{proposition}
	Let $X$ be a random variable following the CTP distribution. The characteristic function of $X$ is defined for all $t \in \mathbb{R}$ par : 
	\begin{equation} 
		\label{eq:cf-CTP}
		\phi_X(t) = \sum_{k=0}^{\infty} \frac{\alpha (it)^k x_0^k \left[ -\delta_1 k^2 + (5\delta_1 + 2\delta_2) k\alpha - 6\alpha^2 \right]}{k! \, (k-\alpha)(k-2\alpha)(k-3\alpha)}, \quad \alpha>k.
	\end{equation}
\end{proposition}

\begin{proof} 
	By replacing $t$ in Equation \eqref{eq:mgf-CTP} with $it$, we obtain Equation \eqref{eq:cf-CTP}.
\end{proof}

\subsection{Reliability and hazard functions}

\begin{proposition}
	The reliability and hazard functions of CTP distribution are respectively given for all $t>0$ by
	\begin{equation}
		R(t) = (3 - 2\delta_1 - \delta_2) \left( \frac{x_0}{x} \right)^{\alpha} + (3\delta_1 + 2\delta_2 - 3) \left( \frac{x_0}{x} \right)^{2\alpha} + (1 - \delta_1 - \delta_2) \left( \frac{x_0}{x} \right)^{3\alpha}
	\end{equation}
	and
	\begin{equation}
		h(t) = \frac{\alpha}{x} \left[ \dfrac{(3 - 2\delta_1 - \delta_2) + (6\delta_1 + 4\delta_2 - 6) \left( \dfrac{x_0}{x} \right)^{\alpha} + (3 - 3\delta_1 - 3\delta_2) \left( \dfrac{x_0}{x} \right)^{2\alpha}}{(3 - 2\delta_1 - \delta_2) + (3\delta_1 + 2\delta_2 - 3) \left( \dfrac{x_0}{x} \right)^{\alpha} + (1 - \delta_1 - \delta_2) \left( \dfrac{x_0}{x} \right)^{2\alpha}} \right].
	\end{equation}
\end{proposition}

\begin{proof}
	The proof is simple recalling that for all $t>0$, the reliability function is given by $R(t) = 1-F(t)$ and the hazard function is given by 
	$$
	h(t) = \frac{f(t)}{1-F(t)} = \frac{f(t)}{R(t)}.
	$$
\end{proof}

\section{Parameter estimation}
\label{sec5}

Let $x_1, \ldots, x_n$ be a random sample of size $n$ from the Cubic Transmuted Pareto defined by pdf \eqref{eq:pdf-CTP}. The general form of the log-likelihood of a CTP is
\begin{equation}
	\ell = \sum_{i=1}^n \log f(x_i)
\end{equation}
which can be rewritten as
\begin{multline}
	\ell = n \log(\alpha) + n\alpha \log(x_0) - (\alpha+1) \sum_{i=1}^n \log(x_i) \\
	+ \sum_{i=1}^n \log \left[ (3 - 2\delta_1 - \delta_2) + (6\delta_1 + 4\delta_2 - 6) \left( \frac{x_0}{x_i} \right)^{\alpha} + (3 - 3\delta_1 - 3\delta_2) \left( \frac{x_0}{x_i} \right)^{2\alpha} \right].
\end{multline}

Since for all $i=1,\ldots,n$, $x_0 \leqslant x_i$, the maximum likelihood estimator (MLE) of $x_0$ is the minimum of the sample values i.e. $\min_{1 \leqslant i \leqslant n} x_i$. The MLEs of the other parameters are obtained by maximizing the respective log-likelihoods of the different models defined below.

\begin{list}{$\bullet$}{}
	\item The log-likelihood corresponding to both the $CTP_G$ and the $CTP_{MG}$ is
	\begin{multline} 
		\ell_G(\alpha, \lambda_1, \lambda_2) = n \log(\alpha) + n\alpha \log(x_0) - (\alpha+1) \sum_{i=1}^n \log(x_i) \\
		+ \sum_{i=1}^n \log \left[ (3 - \lambda_1 - \lambda_2) + (2\lambda_1 + 4\lambda_2 - 6) \left( \frac{x_0}{x_i} \right)^{\alpha} + (3 - 3\lambda_2) \left( \frac{x_0}{x_i} \right)^{2\alpha} \right].
	\end{multline}
	The maximum likelihood estimate (MLE) $(\hat\alpha, \hat\lambda_1, \hat\lambda_2)$ of $(\alpha, \lambda_1, \lambda_2)$ is the solution to the constrained optimization problem
	\begin{equation}
		(\hat\alpha, \hat\lambda_1, \hat\lambda_2) = \argmax_{(\alpha, \lambda_1, \lambda_2) \in \reels_+^* \times \mathscr{S}} \ell_G(\alpha, \lambda_1, \lambda_2),
	\end{equation}
	where $\mathscr{S} = \mathscr{S}_{G}$ if we consider the $CTP_G$ and $\mathscr{S} = \mathscr{S}_{MG}$ if we consider the $CTP_{MG}$.
	
	\item The log-likelihood corresponding to both the $CTP_A$ of \citet{ansari-eledum-2018} and the $CTP_{MA}$ is
	\begin{multline}
		\ell_A(\alpha, \lambda) = n \log(\alpha) + n\alpha \log(x_0) - (\alpha+1) \sum_{i=1}^n \log(x_i) \\
		+ \sum_{i=1}^n \log \left[ 1 - 2\lambda \left( \frac{x_0}{x_i} \right)^{\alpha} + 3\lambda \left( \frac{x_0}{x_i} \right)^{2\alpha} \right].
	\end{multline}
	The MLE $(\hat\alpha, \hat\lambda)$ of $(\alpha,\lambda)$ is the solution to the constrained optimization problem
	\begin{equation}
		(\hat\alpha, \hat\lambda) = \argmax_{(\alpha, \lambda) \in \reels_+^* \times \mathscr{S}} \ell_A(\alpha, \lambda),
	\end{equation}
	where $\mathscr{S} = [-1,1]$ if we consider the $CTP_A$ and $\mathscr{S} = [-1,3]$ if we consider the $CTP_{MA}$.
	
	\item The log-likelihood corresponding to both $CTP_{R18a}$ \citep{rahman-et-al-2020a} and $CTP_{MR18a}$ is
	\begin{multline}
		\ell_{R18a}(\alpha, \lambda_1, \lambda_2) = n \log(\alpha) + n\alpha \log(x_0) - (\alpha+1) \sum_{i=1}^n \log(x_i) \\
		+ \sum_{i=1}^n \log \left[ (1 - \lambda_1 - \lambda_2) + 2(\lambda_1 + 2\lambda_2) \left( \frac{x_0}{x_i} \right)^{\alpha} - 3\lambda_2 \left( \frac{x_0}{x_i} \right)^{2\alpha} \right].
	\end{multline}
	The MLE $(\hat\alpha, \hat\lambda_1, \hat\lambda_2)$ of $(\alpha, \lambda_1, \lambda_2)$ is the solution to the constrained optimization problem
	\begin{equation}
		(\hat\alpha, \hat\lambda_1, \hat\lambda_2) = \argmax_{(\alpha, \lambda_1, \lambda_2) \in \reels_+^* \times \mathscr{S}} \ell_{R18a}(\alpha, \lambda_1, \lambda_2),
	\end{equation}
	where $\mathscr{S} = \mathscr{S}_{R18a}$ if we consider the $CTP_{R18a}$ and $\mathscr{S} = \mathscr{S}_{MR18a}$ if we consider the $CTP_{MR18a}$.
	
	\item The log-likelihood corresponding to both $CTP_{R18b}$ of \citet{rahman-et-al-2021} and $CTP_{MR18b}$ is 
	\begin{multline}
		\ell_{R18b}(\alpha, \lambda_1, \lambda_2) = n \log(\alpha) + n\alpha \log(x_0) - (\alpha+1) \sum_{i=1}^n \log(x_i) \\
		+ \sum_{i=1}^n \log \left[ (1 - \lambda_1) + 2(\lambda_1 - \lambda_2) \left( \frac{x_0}{x_i} \right)^{\alpha} + 3\lambda_2 \left( \frac{x_0}{x_i} \right)^{2\alpha} \right].
	\end{multline}
	The MLE $(\hat\alpha, \hat\lambda_1, \hat\lambda_2)$ of $(\alpha, \lambda_1, \lambda_2)$ is the solution to the constrained optimization problem
	\begin{equation}
		(\hat\alpha, \hat\lambda_1, \hat\lambda_2) = \argmax_{(\alpha, \lambda_1, \lambda_2) \in \reels_+^* \times \mathscr{S}} \ell_{R18b}(\alpha, \lambda_1, \lambda_2),
	\end{equation}
	where $\mathscr{S} = \mathscr{S}_{R18b}$ if we consider the $CTP_{R18b}$ and $\mathscr{S} = \mathscr{S}_{MR18b}$ if we consider the $CTP_{MR18b}$.
	
	\item The log-likelihood corresponding to both $CTP_{R19}$ and $CTP_{MR19}$ is
	\begin{multline}
		\ell_{R19}(\alpha, \lambda) = n \log(\alpha) + n\alpha \log(x_0) - (\alpha+1) \sum_{i=1}^n \log(x_i) \\
		+ \sum_{i=1}^n \log \left[ (1-\lambda) + 6\lambda \left( \frac{x_0}{x_i} \right)^{\alpha} - 6\lambda \left( \frac{x_0}{x_i} \right)^{2\alpha} \right].
	\end{multline}
	The MLE $(\hat\alpha, \hat\lambda)$ of $(\alpha, \lambda)$ is the solution to the constrained optimization problem
	\begin{equation}
		(\hat\alpha, \hat\lambda) = \argmax_{(\alpha, \lambda) \in \reels_+^* \times \mathscr{S}} \ell_{R19}(\alpha, \lambda),
	\end{equation}
	where $\mathscr{S} = [-1,1]$ if we consider the $CTP_{R19}$ and $\mathscr{S} = [-2,1]$ if we consider the $CTP_{MR19}$.
	
	\item The log-likelihood corresponding to $CTP_{R23}$ is
	\begin{multline}
		\ell_{R23}(\alpha, \lambda, \eta) = n \log(\alpha) + n\alpha \log(x_0) - (\alpha+1) \sum_{i=1}^n \log(x_i) \\
		+ \sum_{i=1}^n \log \left[ (1 - \lambda) + 2\lambda (1 + \eta) \left( \frac{x_0}{x_i} \right)^{\alpha} - 3 \lambda \eta \left( \frac{x_0}{x_i} \right)^{2\alpha} \right].
	\end{multline}
	The MLE $(\hat\alpha, \hat\lambda, \hat\eta)$ of $(\alpha, \lambda, \eta)$ is the solution to the constrained optimization problem
	\begin{equation}
		(\hat\alpha, \hat\lambda, \hat\eta) = \argmax_{(\alpha, \lambda, \eta) \in \reels_+^* \times [-1,1] \times [0,2]} \ell_{R23}(\alpha, \lambda, \eta).
	\end{equation}
\end{list}

All these constrained optimization problems are difficult to solve analytically and therefore require the use of a numerical optimization algorithm. One key requirement for such algorithms is their ability to handle inequality constraints. In the remainder of this paper, we will use the R function \textbf{constrOptim}, which includes optimization algorithms that incorporate inequality constraints.

\section{Real data analysis}
\label{sec:real-data}

\subsection{Wheaton river dataset}
These data represent the peak flood exceedances (in $\text{m}^3$/s) of the Wheaton River (Canada). They consist of 72 excesses recorded between 1958 and 1984, rounded to the first decimal place. The data, originally provided by \citet{choulakian-stephens-2001}, have been used in studies by \citet{merovci-puka-2014}, \citet{ansari-eledum-2018}, and \citet{rahman-et-al-2020a}. Descriptive statistics for the Wheaton River dataset are presented in Table \ref{tab:wheaton-stats}. 

\begin{table}[!h]
	\centering
	\caption{Descriptive statistics for the Wheaton river dataset}
	\label{tab:wheaton-stats}
	\begin{tabularx}{0.8\textwidth}{*{6}{Y}}
		\hline Min & $Q_1$ & Median & Mean & $Q_3$ & Max \\
		\hline 0.1 & 2.125 & 9.5 & 12.204 & 20.125 & 64 \\
		\hline 
	\end{tabularx}
\end{table} 

\subsubsection{Fitting with the initial (unmodified) distributions}

Table \ref{tab:wheaton-results-1} gives the estimated parameters and Table \ref{tab:wheaton-comp-1} gives the different criteria for unmodified models.

\begin{table}[!h]
	\centering
	\caption{Estimated parameters for unmodified models for the Wheaton river dataset}
	\label{tab:wheaton-results-1}
	\begin{tabularx}{0.8\textwidth}{cYYYY}
		\hline Distributions & \multicolumn{4}{c}{Estimations} \\ 
		\hline $CTP_{G}$ & $x_0 = 0.1$ & $\hat\alpha = 0.48$ & $\hat\lambda_1 = 0.059$ & $\hat\lambda_2 = -1$ \\
		$CTP_{A}$ & $x_0 = 0.1$ & $\hat\alpha = 0.256$ & $\hat\lambda = -0.934$ & \\
		$CTP_{R18a}$ & $x_0 = 0.1$ & $\hat\alpha = 0.426$ & $\hat\lambda_1 = -0.95$ & $\hat\lambda_2 = -1$ \\
		$CTP_{R18b}$ & $x_0 = 0.1$ & $\hat\alpha = 0.355$ & $\hat\lambda_1 = -1$ & $\hat\lambda_2 = 0.077$ \\
		$CTP_{R19}$ & $x_0 = 0.1$ & $\hat\alpha = 0.2$ & $\hat\lambda = 0.949$ & \\
		$CTP_{R23}$ & $x_0 = 0.1$ & $\hat\alpha = 0.198$ & $\hat\lambda = 1$ & $\hat\eta = 1.918$ \\
		TP & $x_0 = 0.1$ & $\hat\alpha = 0.35$ & $\hat\lambda = -0.952$ & \\
		Pareto & $x_0 = 0.1$ & $\hat\alpha = 0.244$ & & \\
		\hline 
	\end{tabularx}
\end{table} 

\begin{table}[!h]
	\centering
	\caption{Criteria for unmodified models on the Wheaton river dataset (ranks in parentheses)}
	\label{tab:wheaton-comp-1}
	\begin{tabularx}{0.8\textwidth}{cYYYY}
		\hline Distributions & $-\log L^{*}$ & AIC & AICC & BIC  \\
		\hline  
		$CTP_{G}$    & $267.716^{(1)}$ & $541.432^{(1)}$ & $541.785^{(1)}$ & $548.262^{(1)}$ \\
		$CTP_{R18a}$ & $276.901^{(2)}$ & $559.802^{(2)}$ & $560.155^{(2)}$ & $566.632^{(2)}$ \\
		$CTP_{R23}$  & $284.811^{(3)}$ & $575.622^{(4)}$ & $575.975^{(4)}$ & $582.452^{(5)}$ \\
		$CTP_{R19}$  & $285.291^{(4)}$ & $574.582^{(3)}$ & $574.756^{(3)}$ & $579.135^{(3)}$ \\
		$CTP_{R18b}$ & $285.722^{(5)}$ & $577.444^{(6)}$ & $577.797^{(6)}$ & $584.274^{(6)}$ \\
		TP           & $286.201^{(6)}$ & $576.402^{(5)}$ & $576.576^{(5)}$ & $580.955^{(4)}$ \\
		$CTP_{A}$    & $289.828^{(7)}$ & $583.656^{(7)}$ & $583.830^{(7)}$ & $588.209^{(7)}$ \\
		Pareto       & $303.064^{(8)}$ & $608.128^{(8)}$ & $608.185^{(8)}$ & $610.405^{(8)}$ \\ 
		\hline
	\end{tabularx}
\end{table}
According to the results in Table \ref{tab:wheaton-comp-1}, the distribution that best fits the Wheaton River dataset appears to be $CTP_G$. However, as shown in Equation \eqref{eq:range-granz-b}, the $CTP_G$ distribution is not well-defined and does not constitute a valid probability distribution for certain values of $(\lambda_1, \lambda_2) \in \mathscr{S}_G = [0, 1] \times [-1, 1]$. In such cases, the $CTP_G$ cannot be considered a probability distribution in the strict sense. Figure \ref{fig:wheaton-granz} displays the fitted cdf (on the left) and pdf (on the right) from the unmodified Granzotto cubic transmutation ($CTP_G$) for the Wheaton dataset, over the range $x \in [0.1, 0.6]$. We observe that the cdf is negative in the interval $[0.1146777, 0.3690646]$, and the pdf is negative in the interval $[0.1067108, 0.2248255]$.

\begin{figure}[!h]
	\centering
	\includegraphics[scale=0.35]{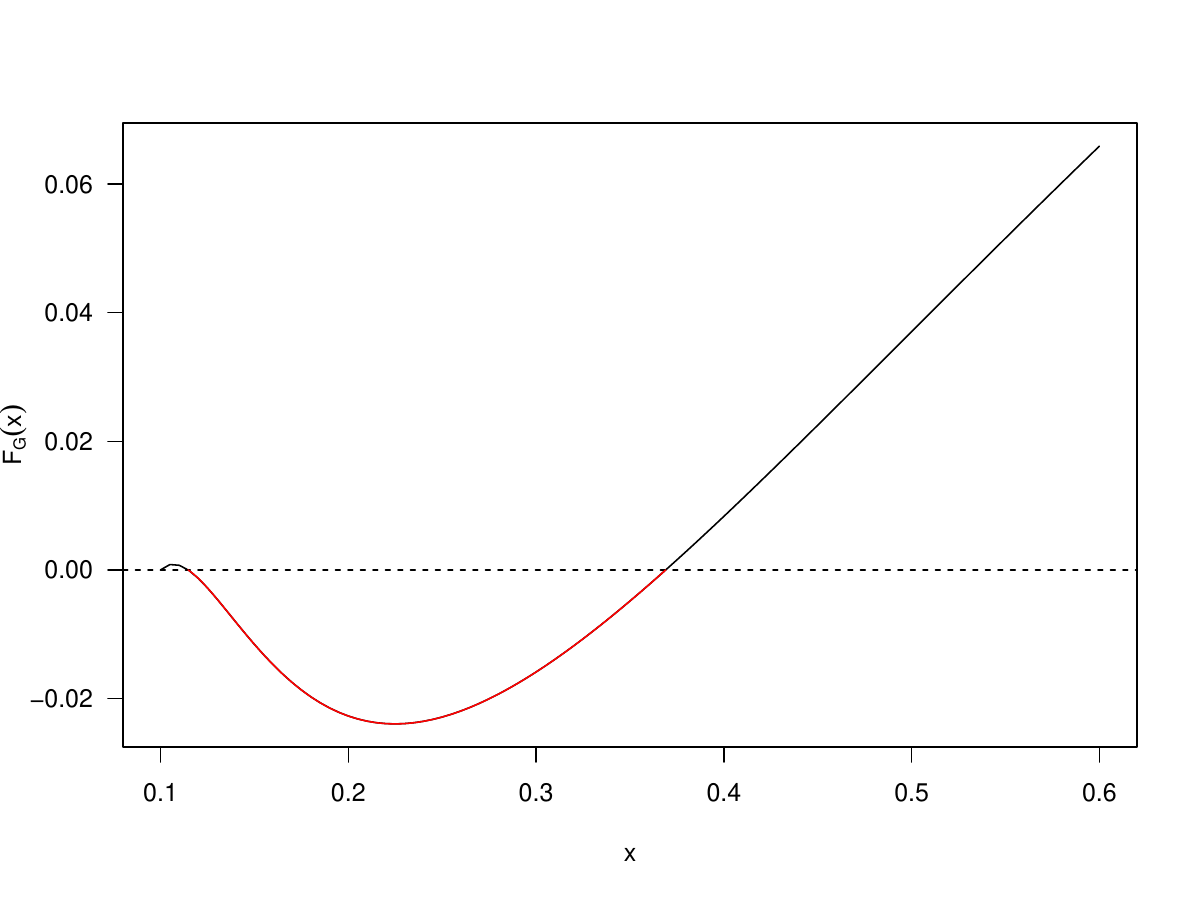}
	\includegraphics[scale=0.35]{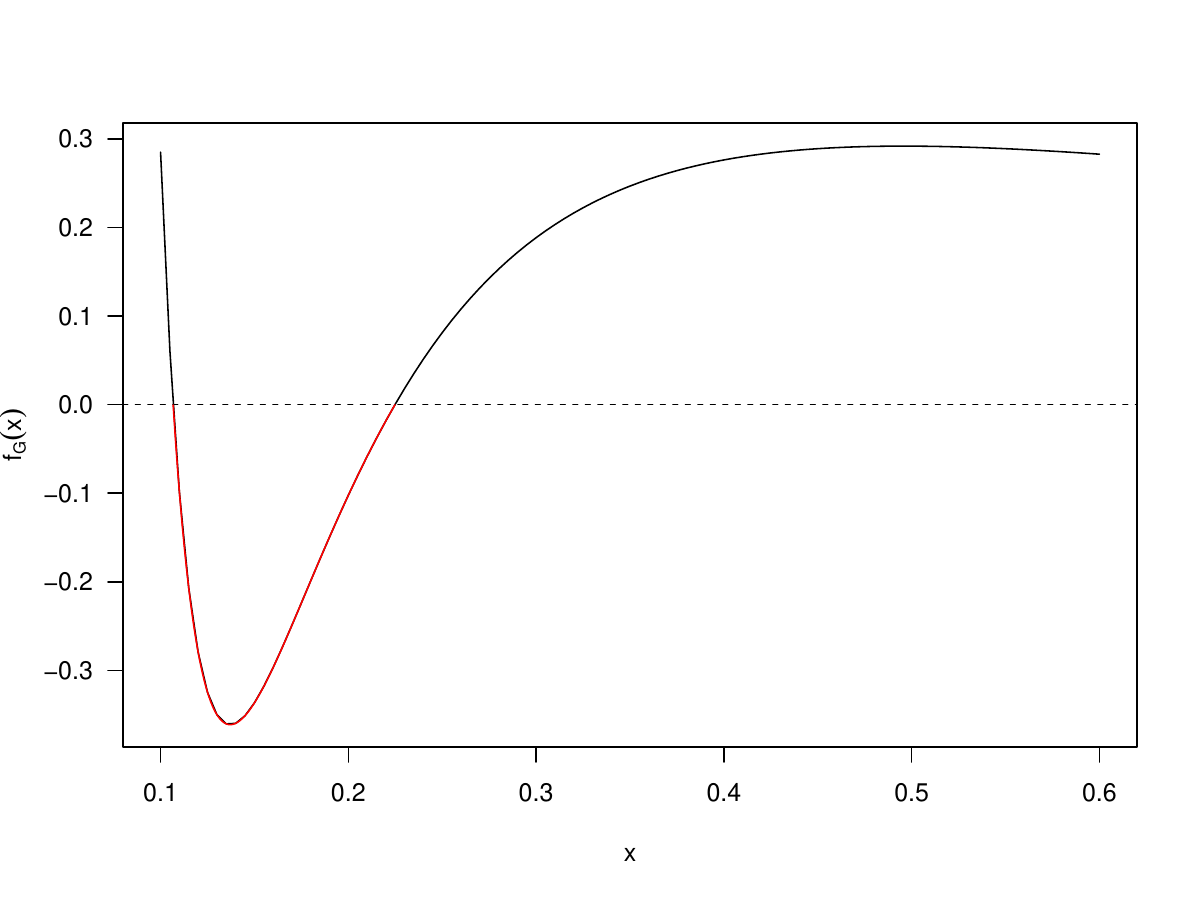}
	\caption{The fitted cdf (left) and pdf (right) from unmodified Granzotto cubic transmutation ($CTP_G$) for Wheaton data}
	\label{fig:wheaton-granz}
\end{figure}

\subsubsection{Fitting with modified distributions}

Table \ref{tab:wheaton-results-2} gives the estimated parameters from the modified models while Table \ref{tab:wheaton-comp-2} gives the different criteria for the modified models.

\begin{table}[!h]
	\centering
	\caption{Estimated parameters for modified models on the Wheaton dataset}
	\label{tab:wheaton-results-2}
	\begin{tabularx}{0.8\textwidth}{cYYcY}
		\hline Distributions & \multicolumn{4}{c}{Estimations} \\ 
		\hline $CTP_{MG}$ & $x_0 = 0.1$ & $\hat\alpha = 0.426$ & $\hat\lambda_1 = 0.05$ & $\hat\lambda_2 = 0$ \\  
		$CTP_{MA}$ & $x_0 = 0.1$ & $\hat\alpha = 0.256$ & $\hat\lambda = -0.934$ & \\              
		$CTP_{MR18a}$ & $x_0 = 0.1$ & $\hat\alpha = 0.426$ & $\hat\lambda_1 = -0.95$ & $\hat\lambda_2 = -1$ \\
		$CTP_{MR18b}$ & $x_0 = 0.1$ & $\hat\alpha = 0.426$ & $\hat\lambda_1 = -1.949$ & $\hat\lambda_2 = 1$ \\
		$CTP_{MR19}$ & $x_0 = 0.1$ & $\hat\alpha = 0.2$ & $\hat\lambda = 0.949$ & \\                   
		$CTP_{M23}$ & $x_0 = 0.1$ & $\hat\alpha = 0.198$ & $\hat\lambda = 1$ & $\hat\eta = 1.918$ \\
		TP & $x_0 = 0.1$ & $\hat\alpha = 0.35$ & $\hat\lambda = -0.952$ & \\
		Pareto & $x_0 = 0.1$ & $\hat\alpha = 0.244$ & & \\
		\hline 
	\end{tabularx}
\end{table} 

\begin{table}[!h]
	\centering
	\caption{Criteria for corrected models on the Wheaton dataset (ranks in parentheses)}
	\label{tab:wheaton-comp-2}
	\begin{tabularx}{0.8\textwidth}{cYYYY}
		\hline Distributions & $-\log L^{*}$ & AIC & AICC & BIC  \\
		\hline  
		$CTP_{MG}$    & $276.901^{(1)}$ & $559.802^{(1)}$ & $560.155^{(1)}$ & $566.632^{(1)}$ \\
		$CTP_{MR18a}$ & $276.901^{(1)}$ & $559.802^{(1)}$ & $560.155^{(1)}$ & $566.632^{(1)}$ \\
		$CTP_{MR18b}$ & $276.901^{(1)}$ & $559.802^{(1)}$ & $560.155^{(1)}$ & $566.632^{(1)}$ \\
		$CTP_{R23}$   & $284.811^{(4)}$ & $575.622^{(5)}$ & $575.975^{(5)}$ & $582.452^{(6)}$ \\
		$CTP_{MR19}$  & $285.291^{(5)}$ & $574.582^{(4)}$ & $574.756^{(4)}$ & $579.135^{(4)}$ \\
		TP            & $286.201^{(6)}$ & $576.402^{(6)}$ & $576.576^{(6)}$ & $580.955^{(5)}$ \\
		$CTP_{MA}$    & $289.828^{(7)}$ & $583.656^{(7)}$ & $583.830^{(7)}$ & $588.209^{(7)}$ \\
		Pareto        & $303.064^{(8)}$ & $608.128^{(8)}$ & $608.185^{(8)}$ & $610.405^{(8)}$ \\
		\hline
	\end{tabularx}
\end{table}


A comparison between Tables \ref{tab:wheaton-results-1} and \ref{tab:wheaton-results-2} shows changes in the parameters of the $CTP_{G}$ and $CTP_{R18b}$ distributions. The best-fitting distributions for the data are now the $CTP_{MG}$, $CTP_{MR18a}$, and $CTP_{MR18b}$. The corrections introduced in this paper have improved the $CTP_{G}$, ensuring that it is well-defined as a probability distribution, and have elevated the $CTP_{R18b}$ from 5th place to a tie for 1st.

Figure \ref{fig:wheaton} presents the histogram of the Wheaton dataset alongside the fitted modified models.

\begin{figure}[!h]
	\centering
	\includegraphics[scale=0.5]{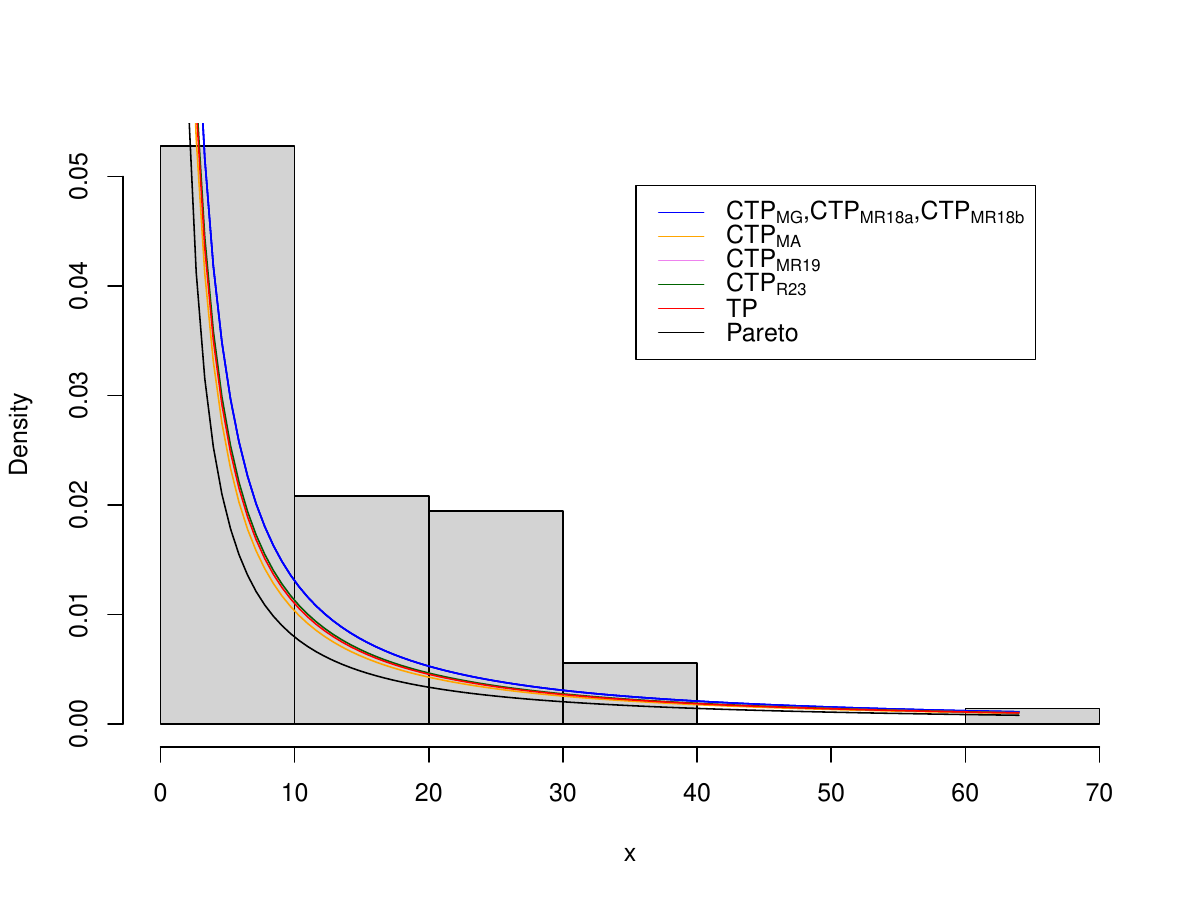}
	\caption{Graphical comparison of modified CTP pdfs fitted to the Wheaton dataset}
	\label{fig:wheaton}
\end{figure} 

\subsection{Norwegian fire insurance data}
%

The dataset consists of fire insurance claims (in 1,000 Norwegian Krone (NOK)) from a Norwegian insurance company over the period 1972–1992. It is available in the R package \textbf{ReIns} by \citet{reynkens-verbelen-2024}. Table \ref{tab:ranks-norwegianfire} presents, for each year, the ranking (from best to worst) of the modified models based on the negative log-likelihood for the fit of the Norwegian fire data.

\begin{table}[!h]
	\small
	\centering
	\caption{Ranking (best to worst) of models based on the opposite of the log-likelihood for Norwegian fire data}
	\label{tab:ranks-norwegianfire}
	\begin{tabularx}{\textwidth}{*{9}{Y}}
		\hline \rotatebox{45}{Year} & \rotatebox{45}{$CTP_{MG}$} & \rotatebox{45}{$CTP_{MA}$} & \rotatebox{45}{$CTP_{MR18a}$} & \rotatebox{45}{$CTP_{MR18b}$} & \rotatebox{45}{$CTP_{MR19}$} & \rotatebox{45}{$CTP_{R23}$} & \rotatebox{45}{TP} & \rotatebox{45}{P} \\ 
		\hline 
		1972 & 1 & 6 & 1 & 1 & 4 & 4 & 7 & 8 \\ 
		1973 & 1 & 6 & 1 & 1 & 5 & 1 & 7 & 8 \\ 
		1974 & 1 & 6 & 1 & 1 & 5 & 1 & 7 & 8 \\ 
		1975 & 1 & 6 & 1 & 1 & 1 & 1 & 7 & 8 \\ 
		1976 & 2 & 6 & 2 & 2 & 5 & 1 & 7 & 8 \\ 
		1977 & 1 & 6 & 1 & 1 & 5 & 1 & 7 & 8 \\ 
		1978 & 2 & 6 & 2 & 2 & 5 & 1 & 7 & 8 \\ 
		1979 & 1 & 7 & 1 & 1 & 6 & 1 & 5 & 8 \\ 
		1980 & 1 & 6 & 1 & 1 & 5 & 1 & 7 & 8 \\ 
		1981 & 1 & 6 & 1 & 1 & 1 & 1 & 7 & 8 \\ 
		1982 & 1 & 6 & 1 & 1 & 7 & 1 & 5 & 8 \\ 
		1983 & 1 & 6 & 1 & 1 & 7 & 4 & 4 & 8 \\ 
		1984 & 1 & 6 & 1 & 1 & 7 & 1 & 5 & 8 \\ 
		1985 & 1 & 6 & 1 & 1 & 5 & 1 & 7 & 8 \\ 
		1986 & 1 & 7 & 1 & 1 & 6 & 4 & 5 & 8 \\ 
		1987 & 1 & 7 & 1 & 1 & 6 & 4 & 5 & 8 \\ 
		1988 & 1 & 6 & 1 & 1 & 7 & 4 & 5 & 8 \\ 
		1989 & 1 & 7 & 1 & 1 & 6 & 4 & 5 & 8 \\ 
		1990 & 1 & 7 & 1 & 1 & 6 & 4 & 5 & 8 \\ 
		1991 & 1 & 7 & 1 & 1 & 6 & 4 & 5 & 8 \\ 
		1992 & 1 & 6 & 1 & 1 & 7 & 4 & 5 & 8 \\ 
		Sum & 23 & 132 & 23 & 23 & 112 & 48 & 124 & 168 \\ 
		\hline 
	\end{tabularx}
\end{table}


We observe that in almost all years, the $CTP_{MG}$, $CTP_{MR18a}$ and $CTP_{MR18b}$  models share the first place. Additionally, all models outperform the Pareto distribution, and in approximately half of the years, the transmuted Pareto distribution performs better than the $CTP_{MR19}$ and $CTP_{MA}$ distributions.

Figure \ref{fig:norwegianfire} illustrates the fitting of the various modified distributions for selected years.

\begin{figure}[p]
	\centering
	\includegraphics[scale=0.4]{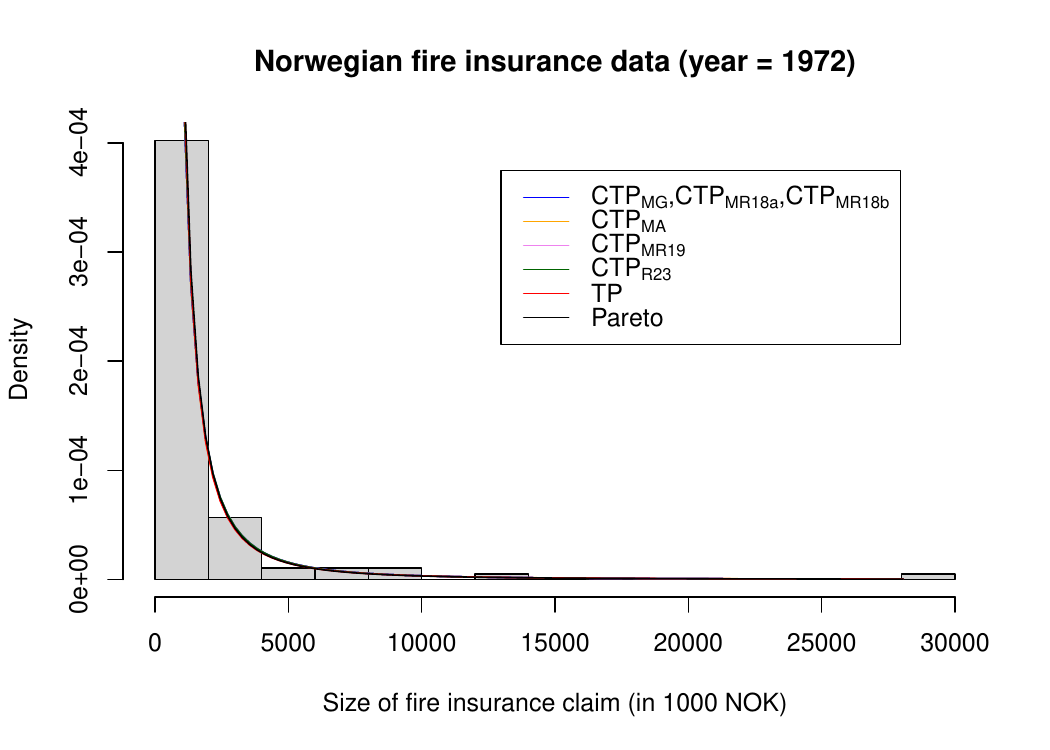}
	\includegraphics[scale=0.4]{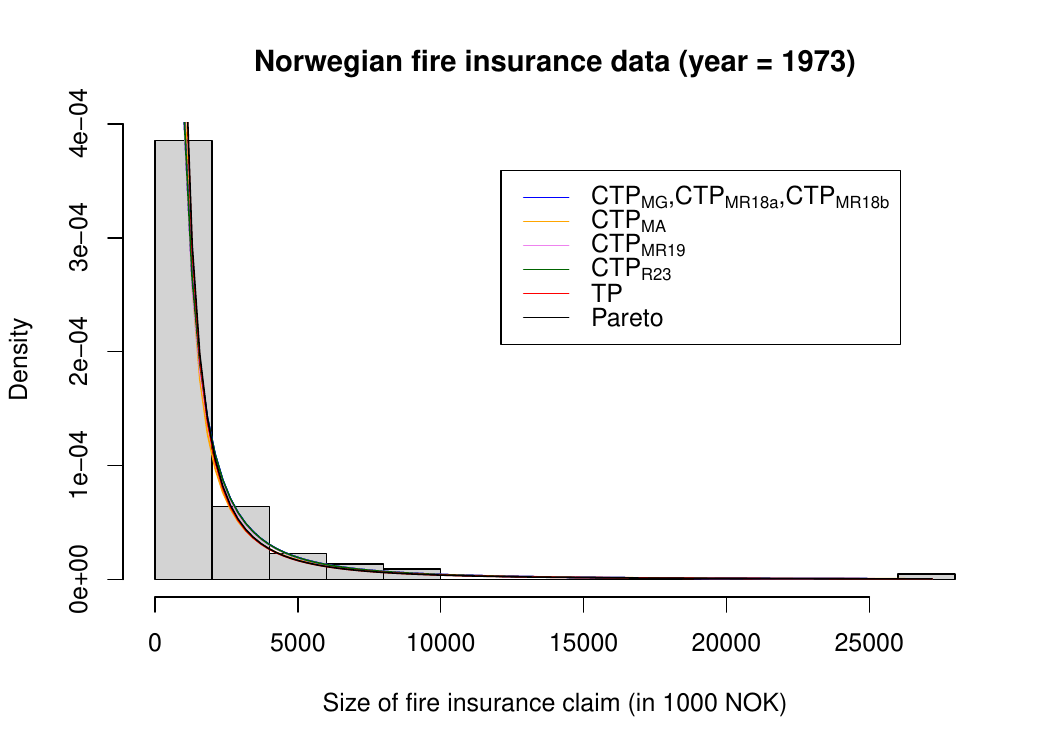}
	\includegraphics[scale=0.4]{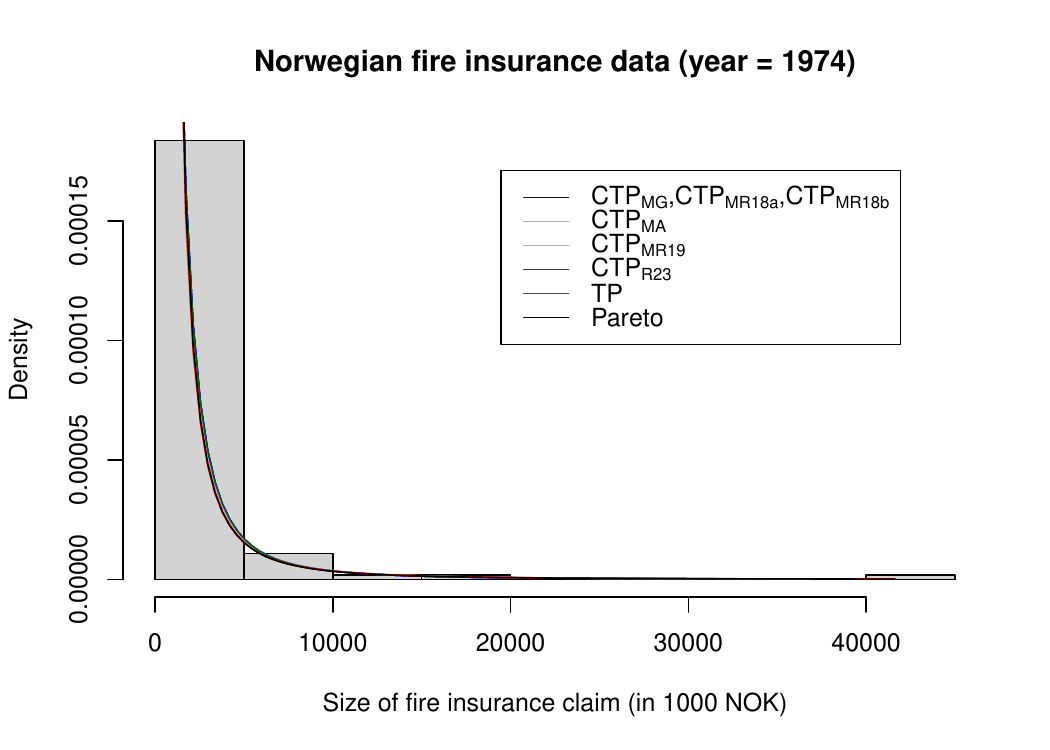}
	\includegraphics[scale=0.4]{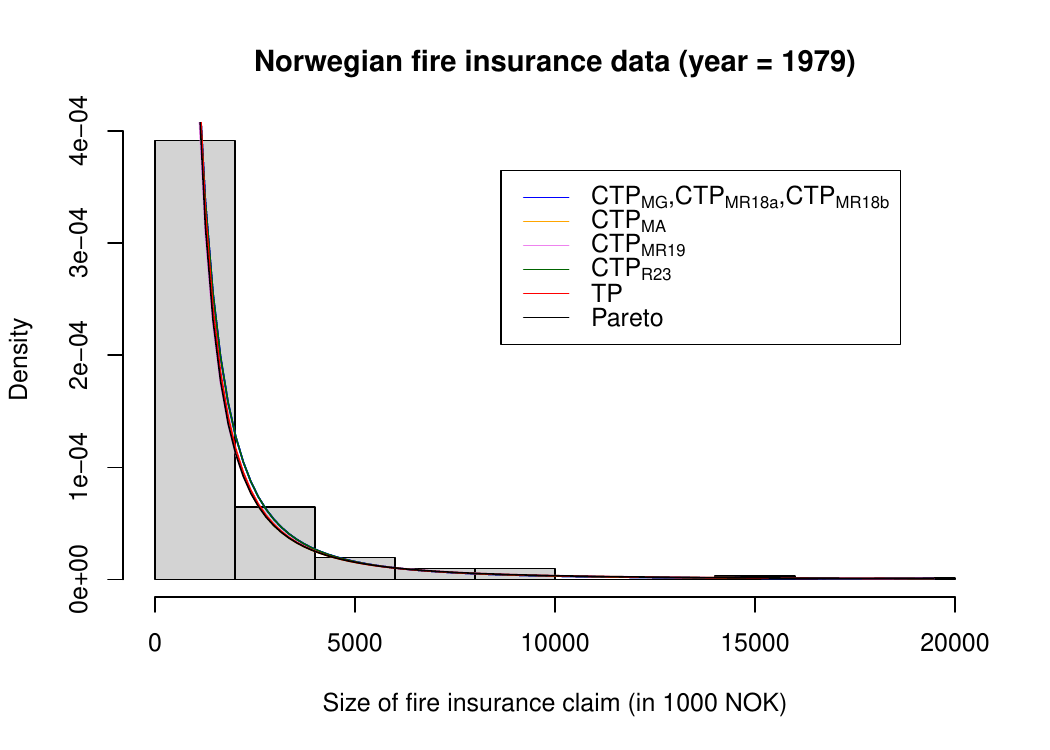}
	\includegraphics[scale=0.4]{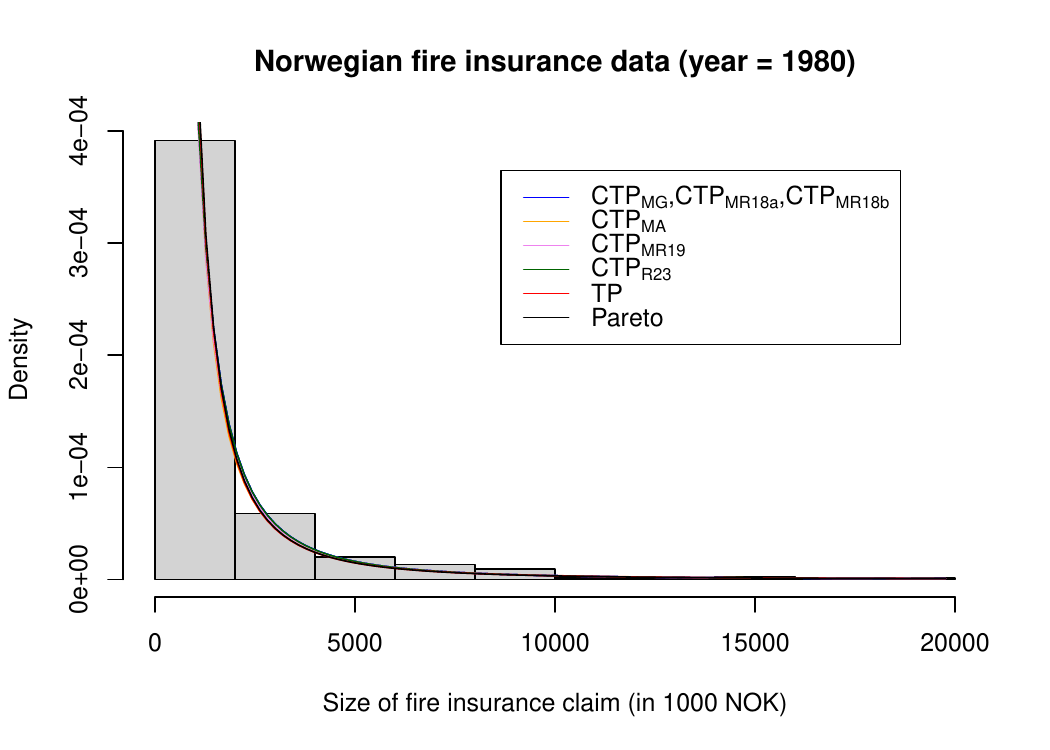}
	\includegraphics[scale=0.4]{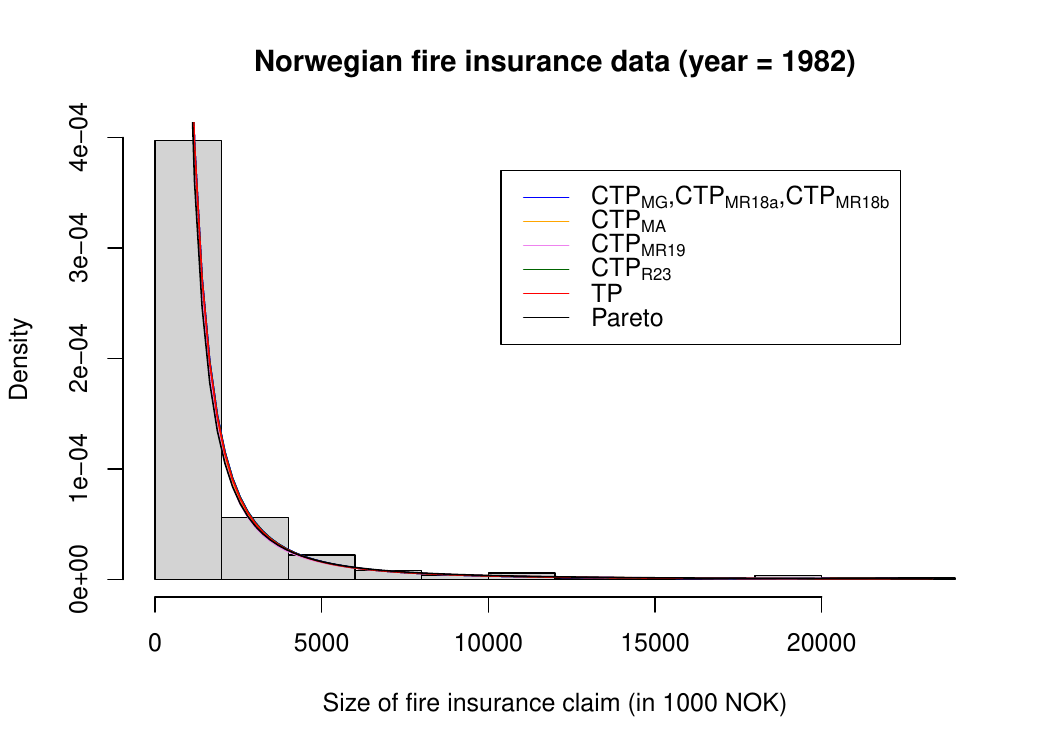}
	\includegraphics[scale=0.4]{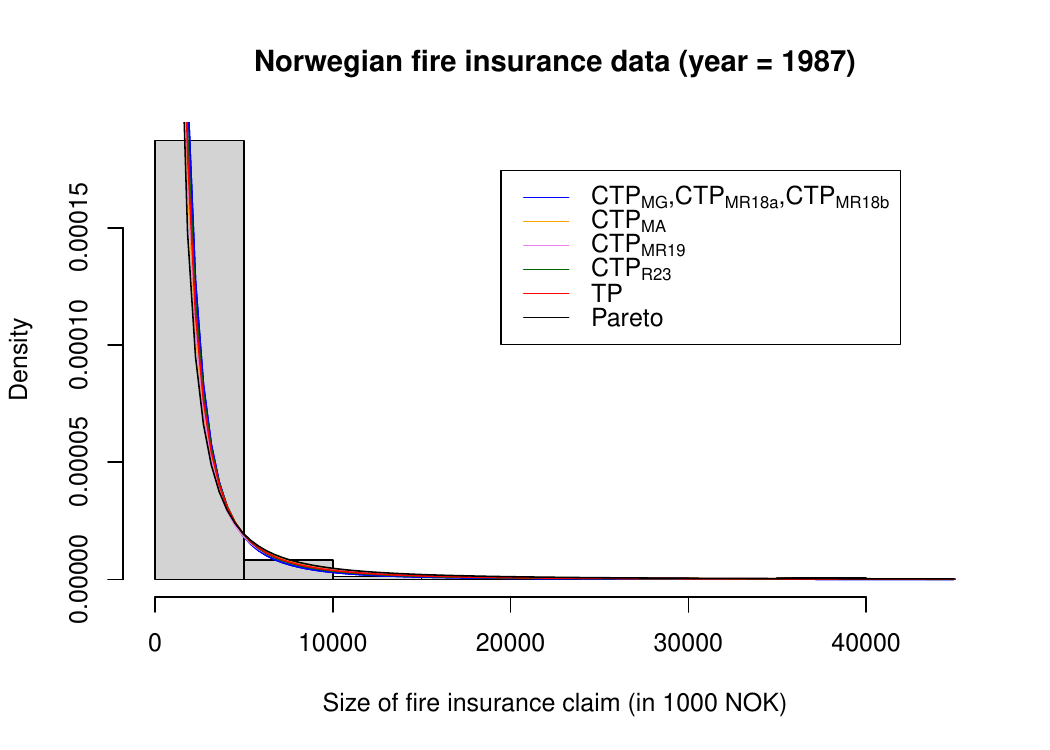}
	\includegraphics[scale=0.4]{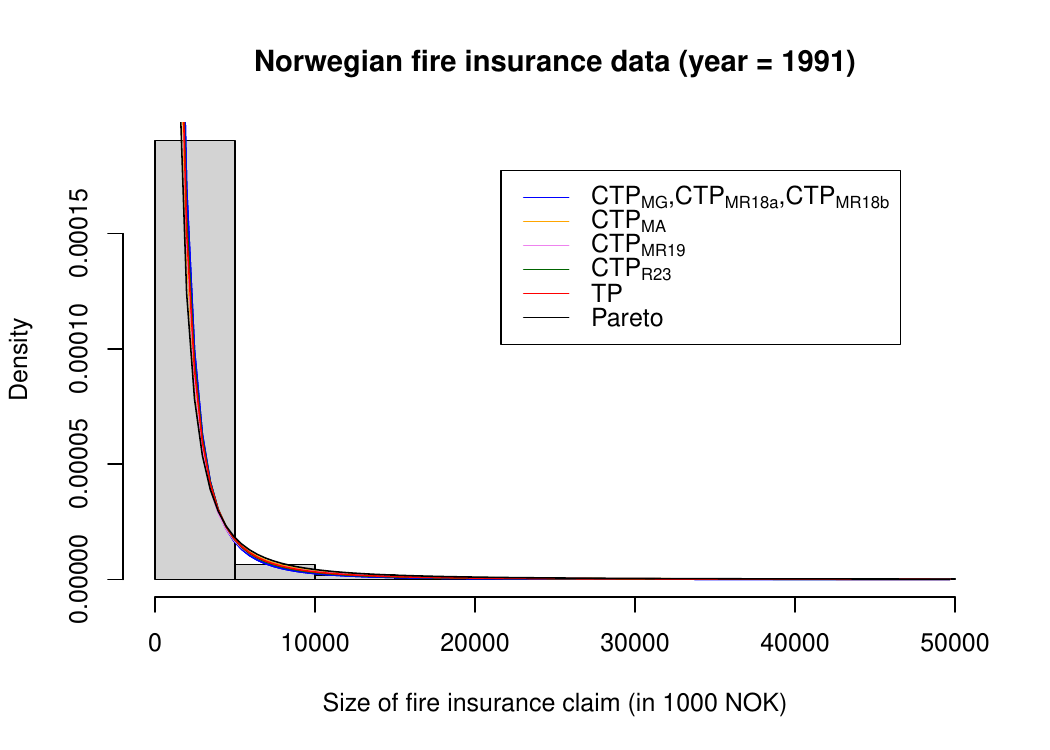}
	\caption{Fitting of different CTP densities to the Norwegian fire insurance data for certain years}
	\label{fig:norwegianfire}
\end{figure}

\section{Conclusion}
\label{sec:conclusion}
%
%

In this paper, we conducted a comprehensive study of six cubic transmutations of the Pareto distribution, introducing a general formulation that unifies and extends both previously studied and novel cubic transmutations. This formulation provides a systematic framework for constructing new transmuted versions while ensuring key probabilistic properties. By analyzing these cubic transmutations theoretically and empirically, we demonstrated their ability to enhance the flexibility of the Pareto distribution, leading to improved fits across diverse real-world datasets, particularly in reliability analysis, economics, and engineering. \\ 

\noindent
Our comparative analysis of the six cubic transmuted Pareto distributions highlighted their respective strengths, with some transmutations offering better fits depending on the dataset. The empirical study, based on the Wheaton River dataset \citep{choulakian-stephens-2001} and Norwegian fire insurance claims \citep{reynkens-verbelen-2024}, showcased the practical relevance of these models. Notably, the refinements proposed by \citet{geraldo-et-al-2025} improved the behavior of some of the transmutations, ensuring they remain well-defined probability distributions and significantly enhancing their performance in capturing extreme values.\\

\noindent
While this work strengthens the theoretical foundation of cubic transmutations, future research could explore further extensions and applications across broader domains. Additionally, investigating the mathematical conditions that guarantee well-defined transmutations, along with a deeper analysis of the trade-off between model complexity and interpretability, would contribute to optimizing their practical use in statistical modeling.

{
	\small 
	\bibliographystyle{apalike}
	\bibliography{paper}
}

\end{document}